\documentclass[aps,prx,showpacs,twocolumn,superscriptaddress,longbibliography]{revtex4-2}

\usepackage{graphicx}
\usepackage{amssymb}
\usepackage{amsmath}
\usepackage{array}
\usepackage{color}
\usepackage{xcolor}
\usepackage{hyperref}
\usepackage{physics}
\usepackage{cancel}
\usepackage{bm}
\usepackage[english]{babel}
\usepackage{times}
\usepackage{amsfonts}
\usepackage{epsfig}
\usepackage{mathrsfs}
\usepackage{color,soul}
\usepackage{bbold}
\usepackage{textcomp}
\usepackage{dsfont}
\usepackage{braket}
\usepackage{empheq}
\usepackage{amsthm}
\usepackage{eurosym}
\usepackage{siunitx}
\AtBeginDocument{\RenewCommandCopy\qty\SI}
\usepackage{enumitem}
\usepackage{comment}
\usepackage{multirow}
\hypersetup{colorlinks=true,linkcolor=blue,citecolor=cyan}
\usepackage{orcidlink}
\usepackage[percent]{overpic}

\def\tr{\operatorname{Tr}}

\def\IP{\mathcal{IP}}

\newcommand{\id}{\mathds{1}}

\newcommand{\ee}{\operatorname{e}}             
\newcommand{\ii}{\mathrm{i}}             

\renewcommand{\braket}[2]{\left\langle #1 \middle| #2 \right\rangle}

\newcommand{\bbra}[1]{{ \langle \! \langle{#1}\vert }}
\newcommand{\kket}[1]{{ \vert {#1}  \rangle \!  \rangle}}

\newcommand{\PProj}[1]{\left[\left[#1\right]\right]}

\newcommand{\cH}{\mathcal{H}}
\newcommand{\cL}{\mathcal{L}}

\newcommand{\cE}{\mathcal{E}}
\newcommand{\cD}{\mathcal{D}}
\newcommand{\cI}{\mathcal{I}}
\newcommand{\cW}{\mathcal{W}}

\newcommand{\cM}{\mathcal{M}}
\newcommand{\D}{\Delta}

\newcommand{\rel}{\mathrm{rel}}
\newcommand{\inc}{\mathrm{inc}}

\newcommand{\NM}{\mathrm{NM}}
\newcommand{\M}{\mathrm{M}}

\newtheorem{definition}{Definition}
\newtheorem{proposition}{Proposition}

\begin{document}

\title{Multi-time advantage for certifying quantum coherence}

\author{Karol Sajnok\orcidlink{0009-0004-5899-8923}}
\email{ksajnok@cft.edu.pl}
\affiliation{Center for Theoretical Physics, Polish Academy of Sciences, Aleja Lotników 32/46, 02-668 Warsaw, Poland}
\affiliation{Nordita, Stockholm University and KTH Royal Institute of Technology, Hannes Alfvéns väg 12, 106 91 Stockholm, Sweden}
\affiliation{Institute of Theoretical Physics, University of Warsaw, ul. Pasteura 5, 02-093 Warsaw, Poland}

\author{Fabio Costa\orcidlink{0000-0002-6547-6005}}
\email{fabio.costa@su.se}
\affiliation{Nordita, Stockholm University and KTH Royal Institute of Technology, Hannes Alfvéns väg 12, 106 91 Stockholm, Sweden}
\affiliation{School of Mathematics and Physics, The University of Queensland, St Lucia, QLD 4072, Australia}

\date{\today}

\begin{abstract}
Non-classical dynamics in open quantum systems can be hidden in temporal correlations that are invisible to standard state- and channel-based tests. We define coherence for multi-time quantum processes and show that access to intermediate times can reveal coherence absent from all corresponding pairwise channels. Incoherent processes are fixed points of dephasing on all input and output spaces, and their coherence is quantified by relative entropy. As a proof of principle, we use an Ising model in which all three pairwise channels are incoherent while the full three-time process remains coherent. We observe this multi-time coherence on a superconducting quantum processor using a coherence witness and full process tomography.
\end{abstract}

\maketitle

\section{Introduction}

A central question in quantum experiments is whether the observed dynamics admits a classical description. Bell tests provide a strong test of non-classicality, but require demanding experimental conditions \cite{brunner2014bell,hensen2015loophole}. Quantum coherence offers a simpler, basis-dependent alternative and has a well-developed resource theory for states \cite{baumgratz2014quantifying, streltsov2017colloquium, winter2016operational} and channels \cite{Saxena2020}. These descriptions, however, can miss non-classical structure contained in temporal correlations generated by system--environment interactions.

Standard open-system dynamics is usually described by maps from a fixed initial time to later times \cite{breuer2002theory, Vacchini2024}. This description does not capture interventions at intermediate times. Instead, for a system repeatedly interacting with an environment, a multi-time process assigns probabilities to arbitrary sequences of local operations \cite{Chiribella2008, Chiribella2009, Pollock2018, pollock2018non}. Such a process contains the reduced dynamical maps as particular restrictions, but it can also retain temporal correlations mediated by inaccessible degrees of freedom. Multi-time processes, also known as process tensors, process matrices, or combs, have become in recent years a central tool in the analysis of quantum non-Markovianity \cite{Ringbauer2015,Strathearn2018, Shrapnel2018,Luchnikov2019,Berk2021,Milz2021} and quantum causal structures \cite{costa2016,Giarmatzi2018, Barrett2019,Goswami2020,Barrett2020,Costa2025}.

Can a multi-time test certify coherence when the usual state- and channel-based tests do not? We show that it can: even when every pairwise channel between the same observation times is incoherent, memory can correlate coherent dynamical histories so that the full multi-time process retains coherence, while the corresponding Markovian process does not. Our framework is illustrated schematically in Fig.~\ref{fig:scheme}.

\begin{figure}[t]
    \centering
    \includegraphics[width=\linewidth]{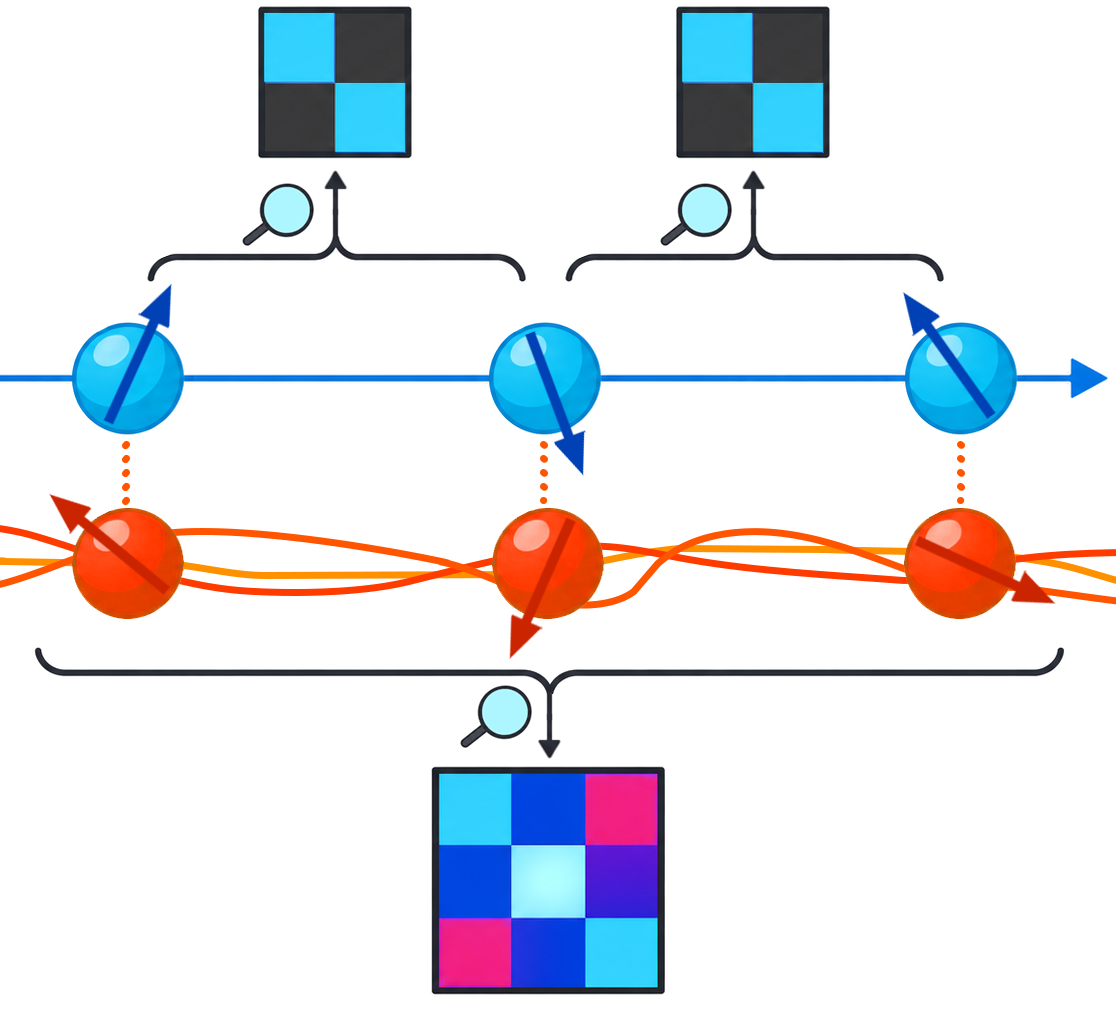}
    \caption{Schematic illustration of multi-time coherence. A persistent memory correlates the system evolution across successive time intervals, so that coherence can be present in the full multi-time process even when the corresponding pairwise channels are incoherent.}
    \label{fig:scheme}
\end{figure}

The framework applies to arbitrary valid processes, including processes with indefinite causal order \cite{Chiribella2013, OCB_2012,Costa2026}. Coherence and indefinite causal order are logically distinct resources: a fixed-order process may be coherent, while an indefinite-order process need not be \cite{Baumeler2014PRA, Baumeler2017}. See also Ref.~\cite{Benhaj2026} for a recent analysis of dephasing in quantum-switch scenarios.

The paper is organized as follows. Section~\ref{sec:background} reviews the process formalism and the state/channel notion of coherence. Section~\ref{sec:process-coherence} defines incoherent multi-time processes, the relative-entropy coherence measure, and coherence witnesses. Section~\ref{sec:ising-model-2} presents the proof of principle through a simple Ising model with two environmental qubits and the comparison with the corresponding pairwise channels. Subsidiary calculations are collected in the Appendix: the Ising model with single environment qubit (\ref{app:ising-model-1}), the dipole--dipole model (\ref{app:dipole-model}) and the details of experimental runs of the main model on the IBM superconductors QPU (\ref{app:ising2-experiment}).

\section{Preliminaries}
\label{sec:background}

Here we review the standard notation and basic definitions of multi-time quantum processes and coherence theory.

\subsection{Multi-time processes}

Consider laboratories $A,B,C,\ldots$, each described by a quantum instrument $\{\cM_a^A\}_a$, $\{\cM_b^B\}_b,\ldots$, where each element is a completely positive (CP) trace-non-increasing map from an input space ($A_I, B_I, \dots$) to an output space ($A_O, B_O, \dots$), either of which may be trivial, and the sum over outcomes is trace preserving (TP).

We first fix the Choi--Jamiołkowski convention used throughout the paper \cite{Jamiolkowski1972,Choi1975}. For a CP map $\cE:\cL(\cH_I)\to\cL(\cH_O)$,
\begin{align}
    J_{\cE}^{IO} &=\sum_{i,j}\ketbra{i}{j}^{I}\otimes\cE(\ketbra{i}{j})^{O}, \label{eq:choi-convention1}\\
    \cE(\rho) &=\tr_I\!\left[(\rho^T\otimes\id_O)J_{\cE}^{IO}\right]. \label{eq:choi-convention2}
\end{align}
The Choi operator in Eq.~\eqref{eq:choi-convention1} and the transpose ${}^T$ in Eq.~\eqref{eq:choi-convention2} are always defined with respect to the same reference basis used below to define dephasing. We denote by $M_a^A$ the Choi operator of the CP map $\cM_a^A$, and analogously for the other laboratories. When composing Choi operators we use the link product \cite{Chiribella2009}. If $M^{XY}$ and $N^{YZ}$ share the system $Y$,
\begin{align}
    M*N=\tr_Y\!\left[(M^{T_Y}\otimes\id_Z)(\id_X\otimes N)\right]. \label{eq:link-product}
\end{align}
This is the Choi-space representation of composition.

In the Choi representation, and following the conventions from Ref.~\cite{Costa2026}, the probability of a sequence of outcomes is given by the generalized Born rule
\begin{align}
    p(a,b,\ldots) = \tr\!\left[ W\left(M_a^A\otimes M_b^B \otimes\cdots\right)^T \right],  \label{eq:generalized-born}
\end{align}
where $W\ge0$ is the process operator, which satisfies the usual linear normalization constraints ensuring unit probability for arbitrary local CPTP maps \cite{Chiribella2009,OCB_2012,araujo2015witnessing,Wechs_2019}, while additional constraints enforce a particular causal order between the labs. For a causally ordered tripartite process $W^{A_OB_IB_OC_I}$ with $A\prec B\prec C$, for example,
\begin{align}
    \tr_{C_I}W^{A_OB_IB_OC_I} &= \id_{B_O}\otimes W^{A_OB_I},\\
    \tr_{B_I}W^{A_OB_I}&=\id_{A_O}.
\end{align}
Here $A$ has only an output, $C$ only an input, and the single intermediate laboratory $B$ constitutes the slot in the conventional terminology.

The same process can also be represented as an auxiliary channel. Group all output spaces into $O$ and all input spaces into $I$, and reorder the tensor factors of $W$ as $W^{OI}$. The associated map $\cW:\cL(\cH_O)\rightarrow\cL(\cH_I)$ acts as
\begin{align}
    \cW(\rho^O) = \tr_O\!\left[ \left((\rho^O)^T\otimes\id^I\right)W^{OI} \right],   \label{eq:process-channel-representation}
\end{align}
which is trace preserving and, for causally ordered processes, coincides with the description of a channel with memory originally introduced in Ref.~\cite{Kretschmann2005}. Up to the fixed reordering of tensor factors, the process operator is the Choi operator of the auxiliary channel $\cW$. This representation of the full process should not be confused with a reduced physical dynamical channel connecting two particular times. For $W^{A_OB_IB_OC_I}$ it maps independently chosen inputs on $A_O B_O$ to the joint outputs on $B_I C_I$.

The classical analogue is obtained by replacing each quantum instrument by conditional probabilities $P_A(a,o_A|i_A)$, $P_B(b,o_B|i_B)$, and the process by a classical stochastic kernel $P_W(\bm i|\bm o)$. The generalized Born rule becomes \cite{Baumeler2017}
\begin{align}
    p(a,b,\ldots) = \sum_{\bm i,\bm o} P_W(\bm i|\bm o) P_A(a,o_A|i_A) P_B(b,o_B|i_B)\cdots . \label{eq:classical-born}
\end{align}
This classical representation is the operational reference against which coherence of a multi-time process will be defined.

\subsection{Coherence of states and channels}

Let $\cH$ be a $d$-dimensional Hilbert space with a fixed orthonormal reference basis $\{\ket{i}\}_{i=0}^{d-1}$. The completely dephasing channel in this basis is
\begin{align}
    \D(\rho)=\sum_{i=0}^{d-1}\ketbra{i}{i}\rho\ketbra{i}{i}.
\end{align}
The set of incoherent states is the fixed-point set $\cI=\{\delta:\D(\delta)=\delta\}$. The relative entropy of coherence is \cite{baumgratz2014quantifying}
\begin{align}
    C_{\rel}(\rho) = S(\rho\Vert\D(\rho)) = S(\D(\rho))-S(\rho), \label{eq:state-crel}
\end{align}
where $S(\rho)=-\tr(\rho\ln\rho)$ is the von Neumann entropy. For a trace-preserving map $\tr J_{\cE}=d_I$, so $\rho_{\cE}=\frac{J_{\cE}}{d_I}$ is a normalized Choi state.

Coherence has been extended from states to quantum operations through several inequivalent resource-theoretic constructions, including coherence capacities and simulation costs \cite{BenDana2017}, coherence-detection power \cite{Theurer2019}, and dynamical coherence of channels \cite{Saxena2020}. Here we adopt the free channel set of Ref.~\cite{Saxena2020}, whose elements neither detect nor generate coherence and satisfy
\begin{align}
    \cE=\D_O\circ\cE\circ\D_I. \label{eq:incoherent-channel}
\end{align}
Channels of this type admit a classical interpretation in the reference basis, in the sense that they are fully characterised by the stochastic map $P_{\mathcal{E}}(i|j):=\expval{\mathcal{E}(\ketbra{j})}{i}$, $\mathcal{E}(\rho) = \sum_{ij} P_{\mathcal{E}}(i|j) \ketbra{i} \expval{\rho}{j}$. In the Choi representation, Eq.~\eqref{eq:incoherent-channel} is equivalent to
\begin{align}
    J_{\cE}=(\D_I\otimes\D_O)(J_{\cE}). \label{eq:incoherent-channel-choi}
\end{align}
A channel is incoherent if and only if its Choi state is incoherent in the corresponding product reference basis.

For our purposes, we quantify this resource by the relative entropy of coherence of the normalized Choi operator,
\begin{align}
    C_{\rel}(\cE) \coloneqq C_{\rel}(\rho_{\cE}) =S(\D(\rho_{\cE}))-S(\rho_{\cE}), \label{eq:channel-coherence}
\end{align}
where $\D=\D_I\otimes\D_O$ acts on both the input and output Choi spaces. This choice differs from the channel-divergence measures considered in Ref.~\cite{Saxena2020}, however admits an immediate extension to general multi-time processes. In particular, a completely dephasing channel is incoherent in this sense. An ordinary channel can equivalently be regarded as a bipartite process (where one party has trivial input and the second a trivial output). We will use this identification below when comparing ordinary channels with genuine multi-time processes.

Previous work has related coherence to the classicality of multi-time measurement statistics \cite{Smirne2019,Strasberg2019,Milz2020}. Here we instead treat coherence as a basis-dependent resource of the full process operator, extending Ref.~\cite{Saxena2020}, and distinguish process coherence from Kolmogorov classicality associated with a fixed probing scheme.

\section{Coherence of multi-time processes} \label{sec:process-coherence}

A process is incoherent if dephasing all local input and output systems does not change any statistics once the local operations are replaced by their corresponding dephased operations. Equivalently, it admits the classical stochastic description of Eq.~\eqref{eq:classical-born}.

\begin{definition}[Process dephasing]
For every open process leg $X$ choose a reference basis and let $\D_X$ denote complete dephasing on that leg. The total process dephasing is
\begin{align}
    \cD(W)=\left(\bigotimes_X\D_X\right)(W). \label{eq:process-dephasing}
\end{align}
\end{definition}
Because every $\D_X$ is CPTP, unital and self-adjoint, $\cD$ preserves positivity and the recursive process normalization constraints. In particular, it maps a valid causally ordered process to a valid process with the same causal order.

\begin{definition}[Incoherent process] \label{def:incoherentprocess}
A finite-dimensional process $W$ is incoherent in the chosen reference basis if
\begin{align}
    W=\cD(W). \label{eq:incoherent-process}
\end{align}
The set of incoherent processes is denoted as $\IP$.
\end{definition}

The operational characterization stated above follows directly from the generalized Born rule. Let $\cD_A=\cD_{A_I}\otimes\cD_{A_O}$ denote dephasing of the Choi operator associated with laboratory $A$, and analogously for the other laboratories. Since dephasing is self-adjoint and, when defined in the same basis as the Choi representation, commutes with transposition,
\begin{align}
    &\tr\!\left[  W  \left( \cD_A(M_a^A)\otimes \cD_B(M_b^B)\otimes\cdots  \right)^T \right] \nonumber\\
    &\qquad= \tr\!\left[ \cD(W) \left(  M_a^A\otimes M_b^B\otimes\cdots \right)^T \right]. \label{eq:dephasing-born-rule}
\end{align}
If $W=\cD(W)$, replacing every local CP map by its dephased counterpart leaves all probabilities unchanged. Conversely, if these probabilities are unchanged for all local CP maps, informational completeness of their Choi operators implies $W=\cD(W)$.

Equivalently, after fixing the process normalization convention, an incoherent process has the diagonal form
\begin{align}
    W_{\inc} = \sum_{\bm i,\bm o} P_W(\bm i|\bm o) \bigotimes_X\ketbra{x_X}{x_X}^{X}, \label{eq:classical-process}
\end{align}
where the diagonal coefficients are the conditional probabilities of a classical multi-time stochastic process and satisfy the corresponding normalization constraints.

Our notion is related to, but more restrictive than, the $K$-classicality of Ref.~\cite{Milz2020}, where classicality is defined through Kolmogorov consistency of statistics obtained from sequential projective measurements in a fixed basis. In the comb representation, this requires that inserting a dephasing channel at unobserved times is indistinguishable from the identity under the remaining fixed-basis projective probes. A classical comb may still contain an off-diagonal contribution that is invisible to these probes, while $W=\mathcal D(W)$ requires the complete process operator itself to be diagonal --- a process may be classical in the sense of Ref.~\cite{Milz2020} and possess nonzero process coherence as defined here. The two notions quantify different resources: statistical nonclassicality under a fixed probing scheme and coherence of the complete multi-time process, respectively.

\begin{definition}[Relative entropy of process coherence]
For $\omega=W/\tr W$, the relative entropy coherence measure is
\begin{align}
    \mathcal{C}_{\rel}(W) \coloneqq S\!\left(\omega\Vert\cD(\omega)\right) = S(\cD(\omega))-S(\omega). \label{eq:process-crel}
\end{align}
\end{definition}

The second equality follows because $\cD$ is an idempotent dephasing map. The definition is also the relative-entropy coherence of the channel representation, Eq.~\eqref{eq:process-channel-representation}, since its normalized Choi operator is $\omega$ up to a permutation of tensor factors. In particular, if $W_{\cE}$ is the bipartite process associated with an ordinary channel $\cE$, then
\begin{align}
    \mathcal{C}_{\rel}(W_{\cE}) = C_{\rel}(\cE). \label{eq:one-slot-process-channel}
\end{align}

\begin{proposition}[Convexity]
\label{prop:convexity}
For valid processes of the same type, the set $\IP$ is convex and the relative entropy of process coherence is convex, i.e. for $0\le p\le1$
\begin{align}
    \mathcal{C}_{\rel}\!\left(pW_1+(1-p)W_2\right) \le p\,\mathcal{C}_{\rel}(W_1) +(1-p)\,\mathcal{C}_{\rel}(W_2). \label{eq:process-coherence-convexity}
\end{align}
\end{proposition}

\begin{proof}
Linearity of $\cD$ immediately implies convexity of $\IP$: if $\cD(W_1)=W_1$ and $\cD(W_2)=W_2$, then
\begin{align}
    \cD\!\left(pW_1+(1-p)W_2\right) = pW_1+(1-p)W_2.
\end{align}
For processes of a fixed type, the normalization constraints fix the trace, so $\tr W_1=\tr W_2$ and the normalized process operator of $W=pW_1+(1-p)W_2$ is $\omega=p\omega_1+(1-p)\omega_2$. Since $\cD$ is linear, joint convexity of quantum relative entropy gives
\begin{align}
    \mathcal{C}_{\rel}(W) &= S\!\left( p\omega_1+(1-p)\omega_2 \,\middle\Vert\, p\cD(\omega_1)+(1-p)\cD(\omega_2) \right) \nonumber\\
    &\le p\,S\!\left(\omega_1\Vert\cD(\omega_1)\right) +(1-p)\,S\!\left(\omega_2\Vert\cD(\omega_2)\right),
\end{align}
which is Eq.~\eqref{eq:process-coherence-convexity}.
\end{proof}

The set $\IP$ is convex by Proposition~\ref{prop:convexity} and closed as the fixed-point set of the continuous linear map $\cD$ within the set of valid processes (which is a closed set, since it includes its boundary). Therefore, whenever $W\notin\IP$, the separating-hyperplane theorem guarantees a coherence witness. 

\begin{definition}[Coherence witness]
With the Choi convention used for higher-order transformations, we call $S=S^\dagger$ a coherence witness if
\begin{align}
    \forall_{W_{\inc}\in\IP}  \tr(S^T W_{\inc})\ge0 \quad \land \quad \exists_W \tr(S^T W)<0, \label{eq:process-witness}
\end{align}
\end{definition}
Such witnesses are experimentally accessible through suitable process measurements, analogously to standard process-matrix witnesses \cite{araujo2015witnessing,Costa2026}.

Following the resource-theoretic frameworks for multi-time processes \cite{Berk2021} or superchannels \cite{Saxena2020}, we take the free transformations to be valid deterministic superprocesses $\Theta$ that are dephasing covariant, $\Theta\circ\mathcal D=\mathcal D\circ\Theta$. This covariance condition is sufficient for monotonicity of Eq.~\eqref{eq:process-crel} by data processing,
\begin{align}
    \mathcal{C}_{\rel}(\Theta(W)) =  S\!\left(\Theta(\omega)\Vert\Theta(\cD(\omega))\right) \le \mathcal{C}_{\rel}(W). \label{eq:disp-monotonicity}
\end{align}

\section{A solvable non-Markovian model} \label{sec:ising-model-2}

We now consider a minimal system--environment model as a proof of principle that multi-time tests can reveal coherence inaccessible to conventional state- and channel-based probes. Importantly, the model admits a classical-memory representation \cite{Giarmatzi2021,Taranto2024,Goswami2025}: a persistent environmental label correlates coherent unitary histories across different time intervals, allowing coherence to survive at the level of the full process even when it is absent from the corresponding reduced channels.

Complementary examples are discussed in Appendices~\ref{app:ising-model-1} and \ref{app:dipole-model}. Appendix~\ref{app:ising-model-1} presents a simpler one-qubit environment construction exhibiting a multi-time advantage over the standard fixed-origin channel description and shows that the residual process coherence grows linearly with the number of observation times. Appendix~\ref{app:dipole-model} considers a physically motivated magnetic dipole--dipole interaction. In this case the environmental memory is genuinely quantum, yet the same stronger separation persists as in our main-text example: all three pairwise channels between three observation times can be incoherent while the full three-time process remains coherent.

Let the system $S$ and two environmental memories $E_1,E_2$ be qubits, and consider
\begin{align}
    H^{SE_1E_2} = 2g\,Z^S Z^{E_1} + g\,Z^S Z^{E_2}, \label{eq:ising2-hamiltonian}
\end{align}
where $g$ is a real parameter. We take the environment to be initially maximally mixed, $\sigma^{E_1E_2} = \frac14\id_{E_1E_2}$. Writing the environmental $Z$ eigenvalues as $s_1,s_2=\pm1$, the joint unitary evolution operator equals
\begin{align}
    U(t)
    =
    \sum_{s_1,s_2=\pm1}
    \ketbra{s_1s_2}{s_1s_2}^{E_1E_2}
    \otimes
    U_{s_1s_2}(t),
    \label{eq:ising2-controlled-unitary}
\end{align}
where $U_{s_1s_2}(t) = \exp\!\left[-\ii gt(2s_1+s_2)Z\right]$. Both environmental labels are conserved by the interaction, so they constitute a classical memory shared by the different temporal intervals \cite{Goswami2025}.

\begin{figure}[t]
    \centering
    \includegraphics[width=\linewidth]{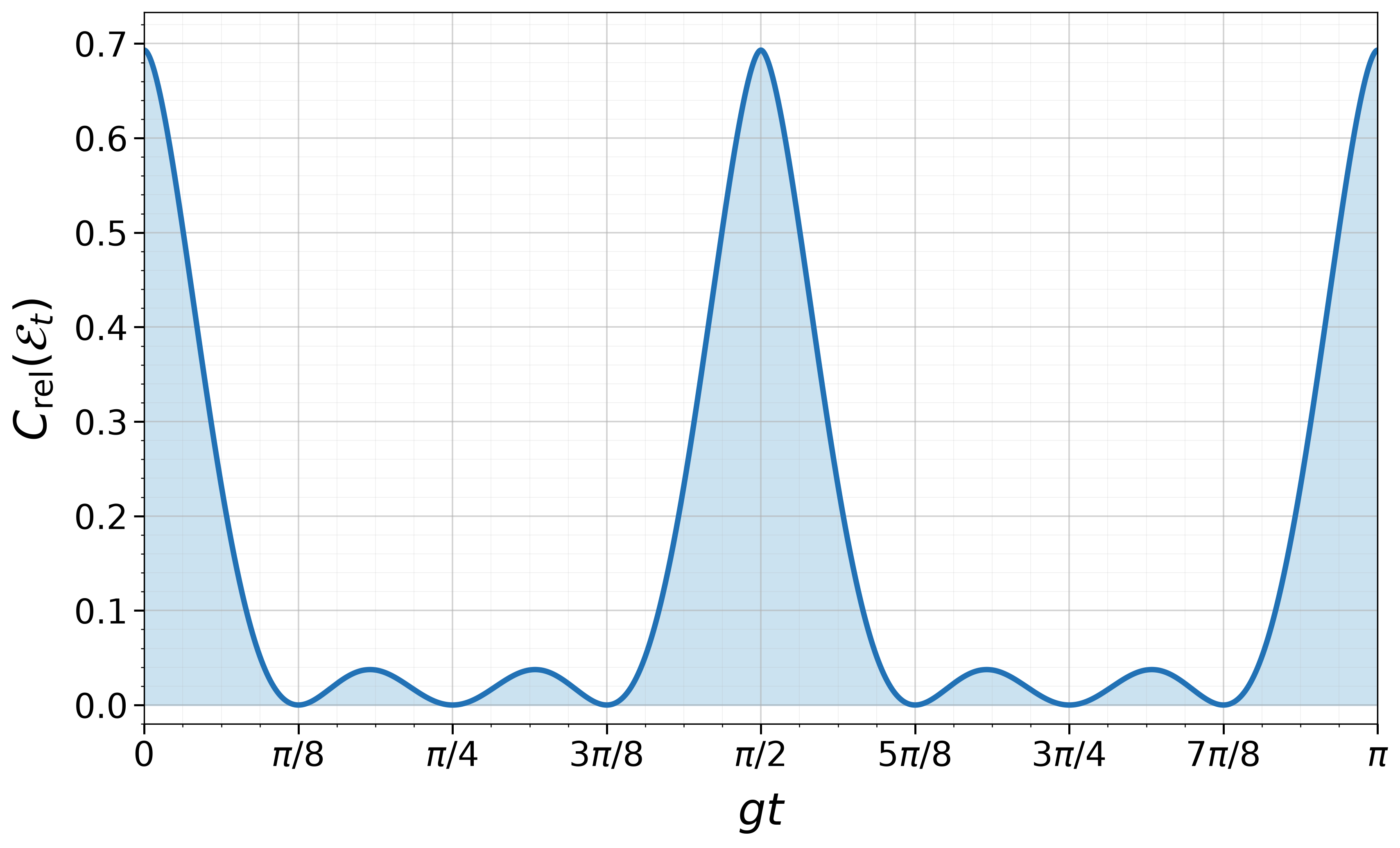}
    \caption{Relative entropy of coherence $C_{\rm rel}(\cE_t)$ of the reduced channel as a function of the dimensionless interaction time $gt$. The coherence vanishes whenever $\gamma(t)=\cos(4gt)\cos(2gt)=0$, corresponding to a completely dephasing reduced channel.}
    \label{fig:ising2_channel_coherence}
\end{figure}

\subsection{Reduced channels}

Tracing over the environment gives
\begin{align}
    \cE_t(\rho) = \frac{1+\gamma(t)}{2}\rho + \frac{1-\gamma(t)}{2}Z\rho Z, \label{eq:ising2-reduced-channel}
\end{align}
where
\begin{align}
    \gamma(t) = \cos(4gt)\cos(2gt). \label{eq:ising2-gamma}
\end{align}
The normalized Choi state has support on $\{\ket{00},\ket{11}\}$ and reads
\begin{align}
    \widetilde J_{\cE_t} =  \frac12 \begin{pmatrix}
        1 & \gamma(t)\\
        \gamma(t) & 1
    \end{pmatrix}. \label{eq:ising2-choi}
\end{align}
Consequently,
\begin{align}
    C_{\rel}(\cE_t) = \ln2 - h_2\!\left[ \frac{1+\gamma(t)}{2} \right], \label{eq:ising2-channel-coherence}
\end{align}
where $h_2(q)=-q\ln q-(1-q)\ln(1-q)$.

For $\tau=\frac{\pi}{8g}$, at both $\tau$ and $2\tau$, $\gamma(\tau)=0=\gamma(2\tau)$, and the channels are fully dephasing, $\cE_\tau = \D = \cE_{2\tau}$. Consider the three observation times $T_0=0$, $T_1=\tau$, and $T_2=2\tau$. The channel $T_0\rightarrow T_1$ is $\cE_\tau$ and the channel $T_0\rightarrow T_2$ is $\cE_{2\tau}$. Moreover, because the environmental populations are conserved and their marginal state remains $\frac14\id$, a causal break at $T_1$ followed by a fresh preparation produces the same reduced channel $\cE_\tau$ during the second interval, $\cE_{T_1\rightarrow T_2} =\D$, and
\begin{align}
    C_{\rel}(\cE_{T_0\rightarrow T_1}) = C_{\rel}(\cE_{T_1\rightarrow T_2})  =0= C_{\rel}(\cE_{T_0\rightarrow T_2}).
    \label{eq:ising2-all-channel-coherences}
\end{align}
The relative entropy of coherence of the reduced channel as a function of interaction time is shown in Fig.~\ref{fig:ising2_channel_coherence}.

\subsection{The three-time process}

The full process retains the same environmental labels across both intervals. For arbitrary interaction times $\tau_1$ and $\tau_2$ it is
\begin{align}
    W_{\NM}(\tau_1,\tau_2) = \frac14 \sum_{s_1,s_2=\pm1} &\PProj{U_{s_1s_2}(\tau_1)}^{A_OB_I} \nonumber\\
    &\otimes \PProj{U_{s_1s_2}(\tau_2)}^{B_OC_I}, \label{eq:ising2-process}
\end{align}
For each Choi pair, define the logical support basis $\ket{0_L}=\ket{00}$, $\ket{1_L}=\ket{11}$. With Eq.~\eqref{eq:ising2-gamma}, introduce
\begin{align}
    \gamma_1&=\gamma(\tau_1), & \gamma_2&=\gamma(\tau_2), & \gamma_\pm&=\gamma(\tau_1\pm\tau_2). \label{eq:ising2-gammas}
\end{align}
In the basis $\{\ket{0_L0_L},\ket{0_L1_L}, \ket{1_L0_L},\ket{1_L1_L}\}$, the normalized process Choi state $\widetilde W_{\NM}=W_{\NM}/\tr W_{\NM}$ is
\begin{align}
    \widetilde W_{\NM}(\tau_1,\tau_2) = \frac14
    \begin{pmatrix}
        1 & \gamma_2 & \gamma_1 & \gamma_+ \\
        \gamma_2 & 1 & \gamma_- & \gamma_1 \\
        \gamma_1 & \gamma_- & 1 & \gamma_2 \\
        \gamma_+ & \gamma_1 & \gamma_2 & 1
    \end{pmatrix}.
    \label{eq:ising2-process-matrix-general}
\end{align}
whose four eigenvalues are
\begin{align}
    \lambda^{+}_\pm &= \frac14\! \left[ 1+\frac{\gamma_++\gamma_-}{2} \pm \frac12 \sqrt{ (\gamma_+-\gamma_-)^2 +4(\gamma_1+\gamma_2)^2 }  \right]\!, \label{eq:ising2-eigenvalues-even}\\
    \lambda^{-}_\pm  &= \frac14 \! \left[ 1-\frac{\gamma_++\gamma_-}{2} \pm \frac12  \sqrt{ (\gamma_+-\gamma_-)^2 +4(\gamma_2-\gamma_1)^2 } \right]\!.
    \label{eq:ising2-eigenvalues-odd}
\end{align}
Since all diagonal entries of the normalized process Choi state are equal, $\cD(\widetilde W_{\NM})  = \frac14\id_4$, and  $S[\cD(\widetilde W_{\NM})] = \ln4$. Finally, the relative entropy of process coherence for arbitrary $\tau_1,\tau_2$ equals
\begin{align}
    \mathcal C_{\rel}  [W_{\NM}(\tau_1,\tau_2)] &= \ln4 + \sum_{\sigma=\pm} \sum_{\eta=\pm}  \lambda^{\sigma}_\eta  \ln\lambda^{\sigma}_\eta. \label{eq:ising2-process-coherence-general}
\end{align}

\begin{figure}[t]
    \centering
    \begin{overpic}[width=\linewidth]{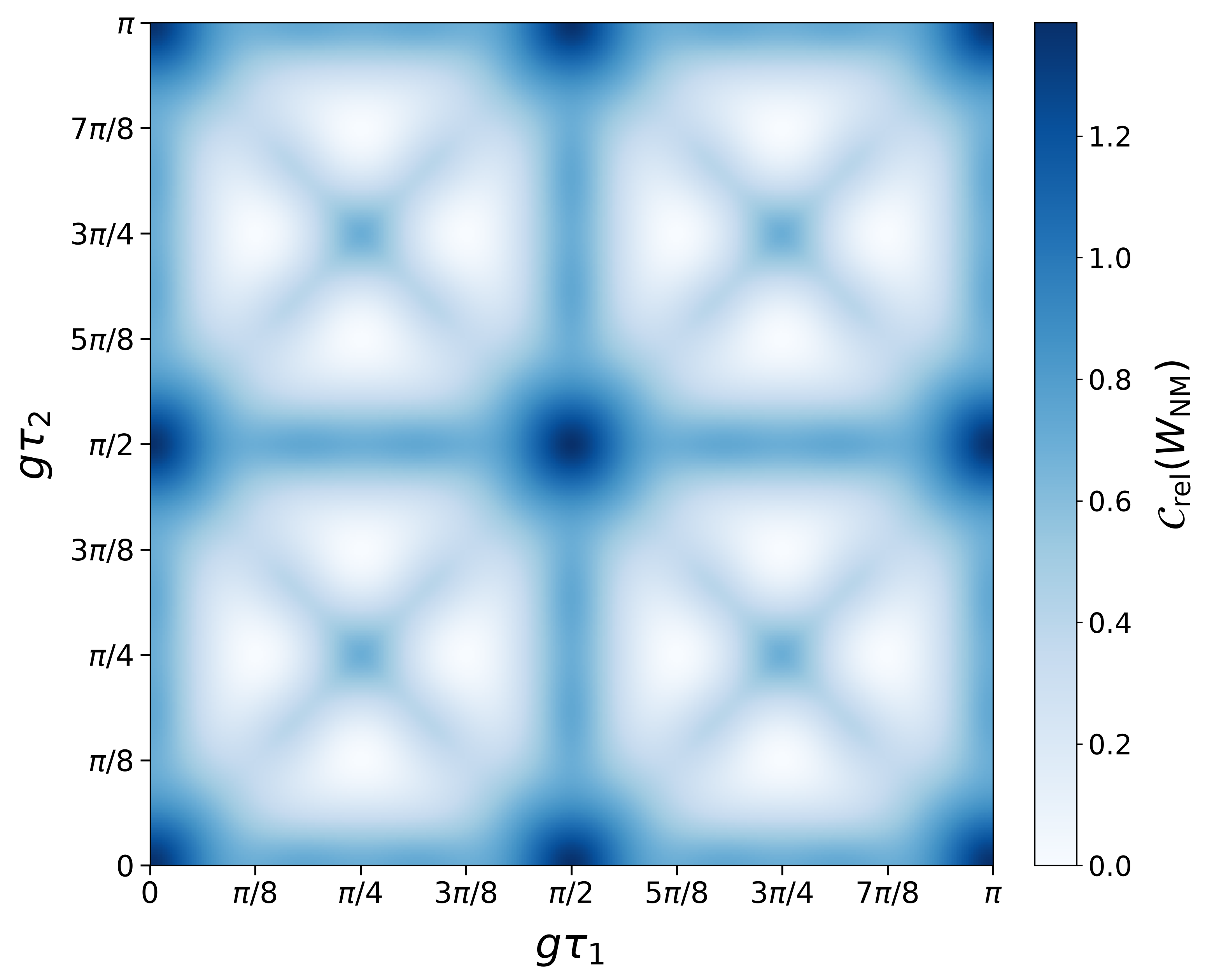}
        \put(1,78){(a)}
    \end{overpic}
    \vspace{0.5em}
    \begin{overpic}[width=\linewidth]{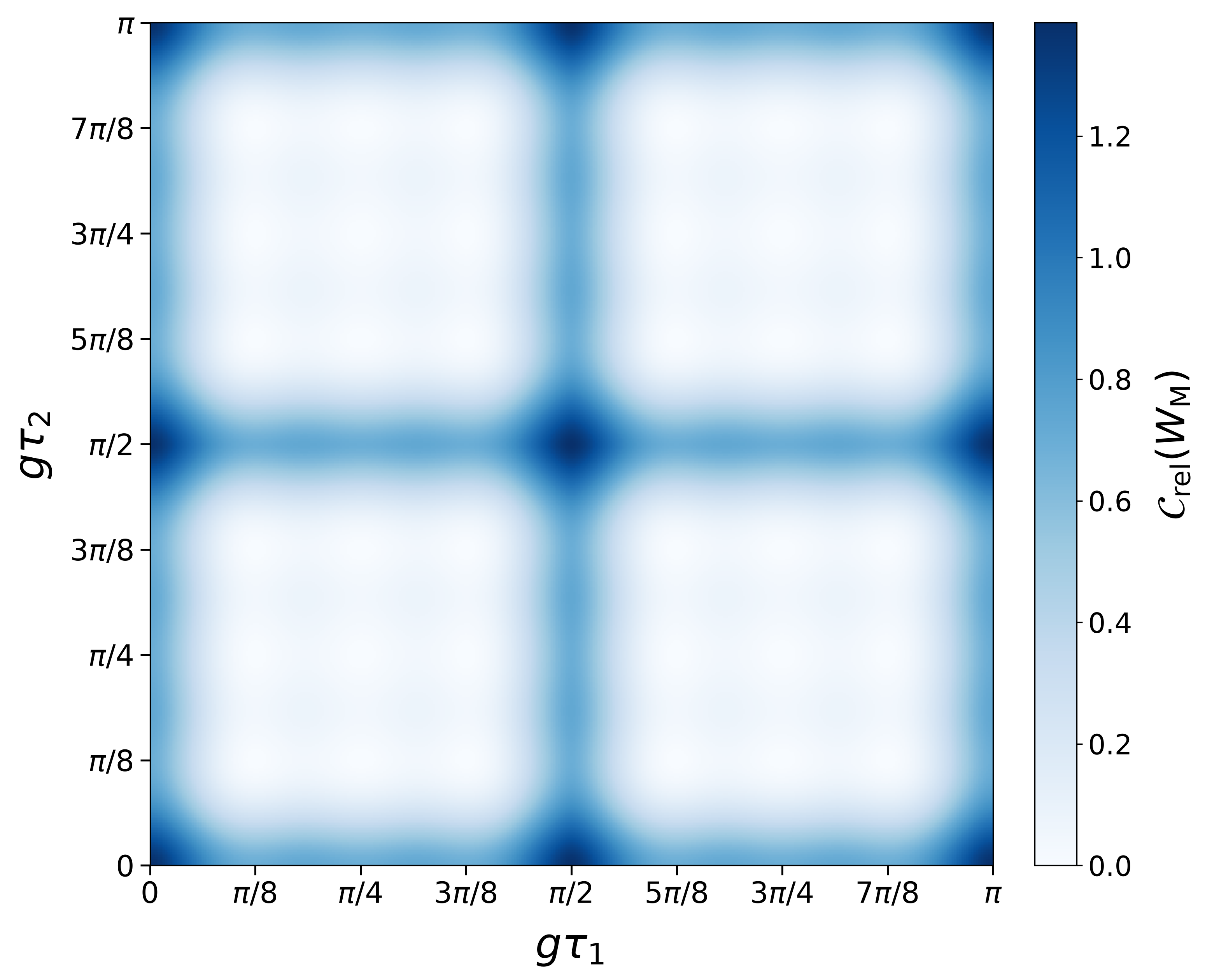}
        \put(1,78){(b)}
    \end{overpic}
    \caption{Relative entropy of coherence (a) for the non-Markovian multi-time process and (b) for the corresponding Markovian process as functions of $\tau_1$ and $\tau_2$. The Markovian coherence vanishes when both reduced channels are incoherent, whereas the non-Markovian process retains nonzero coherence.}
    \label{fig:ising2_multitime_coherence}
\end{figure}

The corresponding Markovian process is obtained by independently averaging over the environmental labels in the two intervals \cite{pollock2018non},
\begin{align}
    W_{\rm M}(\tau_1,\tau_2) = J_{\cE_{\tau_1}}^{A_OB_I} \otimes J_{\cE_{\tau_2}}^{B_OC_I}. \label{eq:ising2-markovian-process}
\end{align}
In the same normalized support representation,
\begin{align}
    \widetilde W_{\rm M}(\tau_1,\tau_2) = \frac14
    \begin{pmatrix}
        1 & \gamma_2 & \gamma_1 & \gamma_1\gamma_2 \\
        \gamma_2 & 1 & \gamma_1\gamma_2 & \gamma_1 \\
        \gamma_1 & \gamma_1\gamma_2 & 1 & \gamma_2 \\
        \gamma_1\gamma_2 & \gamma_1 & \gamma_2 & 1
    \end{pmatrix}.
    \label{eq:ising2-markovian-process-matrix}
\end{align}
Its coherence is additive,
\begin{align}
    \mathcal C_{\rel} [W_{\rm M}(\tau_1,\tau_2)] &=  C_{\rel}(\cE_{\tau_1}) + C_{\rel}(\cE_{\tau_2})  \nonumber\\
    &= 2\ln2  - h_2\!\left(\frac{1+\gamma_1}{2}\right) - h_2\!\left(\frac{1+\gamma_2}{2}\right). \label{eq:ising2-markovian-coherence}
\end{align}
The two matrices in Eqs.~\eqref{eq:ising2-process-matrix-general} and \eqref{eq:ising2-markovian-process-matrix} differ in their correlations between the two temporal intervals: the Markovian process contains the products $\gamma_1\gamma_2$, whereas the non-Markovian process retains the coherent-history factors $\gamma_\pm$. Their relative entropies of coherence are plotted in Fig.~\ref{fig:ising2_multitime_coherence}.

Setting $\tau_1=\tau_2=\tau=\frac{\pi}{8g}$, Eq.~\eqref{eq:ising2-process-matrix-general} reduces to
\begin{align}
    \widetilde W_{\NM}  = \frac14
    \begin{pmatrix}
        1&0&0&0\\
        0&1&1&0\\
        0&1&1&0\\
        0&0&0&1
    \end{pmatrix} \quad \Rightarrow \quad  \mathcal C_{\rel}(W_{\NM}) = \frac12\ln2
    \label{eq:ising2-process-coherence}
\end{align}
This model also illustrates the distinction from the notion of classicality in Ref.~\cite{Milz2020}. Since $[H,Z^S]=0$, sequential measurements in the $Z$ basis satisfy the corresponding Kolmogorov consistency conditions. The surviving off-diagonal temporal correlations in Eq.~\eqref{eq:ising2-process-coherence} are invisible to such fixed-basis probes but can be accessed by a witness.

By contrast, the Markovian process yields
\begin{align}
    \widetilde W_{\rm M} = \frac14\id_4 \quad \Rightarrow \quad \mathcal C_{\rel}(W_{\rm M})=0.\label{eq:ising2-markovian-zero}
\end{align}
The coherence in Eq.~\eqref{eq:ising2-process-coherence} cannot be attributed to either temporal interval individually. It arises from correlations between coherent histories associated with the persistent classical memory.

\subsection{A multi-time coherence witness}

The multi-time coherence in the above non-Markovian process can be detected directly without reconstructing the complete process matrix. Consider a coherence witness (see Section~\ref{sec:process-coherence})
\begin{align}
    S = -X^{A_O}\otimes X^{B_I}\otimes X^{B_O}\otimes X^{C_I}. \label{eq:ising2-witness}
\end{align}
For every incoherent process,
\begin{align}
    \tr(S\widetilde W_{\inc})=0, \label{eq:ising2-witness-incoherent}
\end{align}
since $S$ has no diagonal components in the reference basis, while for Eq.~\eqref{eq:ising2-process-coherence},
\begin{align}
    \sigma \coloneqq  \tr(S\widetilde W_{\NM}) = -\frac12, \label{eq:ising2-witness-value}
\end{align}
certifying coherence of the multi-time process at the chosen interaction times.

The emergence of coherence in the multi-time setting has a direct interpretation in terms of the Hahn spin echo~\cite{Hahn1950} and dynamical-decoupling protocols~\cite{Viola1999,Suter2016}. The conserved labels $(s_1,s_2)$ act as a quasi-static random field, so that averaging over them dephases the reduced dynamics, while an intermediate $\pi$ pulse refocuses the accumulated phases since $XU_{s_1s_2}(\tau)X=U^\dagger_{s_1s_2}(\tau)$ and hence $U_{s_1s_2}(\tau)XU_{s_1s_2}(\tau)=X$. In Eq.~\eqref{eq:ising2-process-matrix-general}, $\gamma_+=\gamma(\tau_1+\tau_2)$ describes the unrefocused contribution, whereas $\gamma_-=\gamma(\tau_1-\tau_2)$ contains the echo contribution, with $\Tr \!\, (S\widetilde W_{\rm NM})=-(\gamma_++\gamma_-)/2$. At $\tau_1=\tau_2=\tau$, $\gamma_+=0$ and $\gamma_-=1$, so the witness signal originates entirely from the refocused sector. In the Markovian process the environmental label is refreshed between the intervals, this contribution factorizes as $\gamma_1\gamma_2$ and vanishes at the same operating point.

For the individual reduced channels, Eq.~\eqref{eq:ising2-choi}, the off-diagonal Choi amplitude can be obtained from two Pauli correlators,
\begin{align}
    \gamma = \langle X\otimes X\rangle - \ii\langle X\otimes Y\rangle. \label{eq:ising2-channel-diagnostic}
\end{align}
Within the dephasing family generated by Eq.~\eqref{eq:ising2-hamiltonian}, $\gamma=0$ is equivalent to vanishing channel coherence.

\subsection{Experimental implementation}
Multi-time processes have previously been characterized on superconducting quantum processors, both for non-Markovian control \cite{White2020} and by full process tomography \cite{Giarmatzi2025}. 

We implement our protocol on IBM superconducting qubits \cite{ibm}. In the convention $R_{ZZ}(\theta)=\exp(-\ii\theta Z\otimes Z/2)$, one interaction interval is
\begin{align}
    U_\tau =  R_{ZZ}^{SE_1}\!\left(\frac{\pi}{2}\right) R_{ZZ}^{SE_2}\!\left(\frac{\pi}{4}\right).
    \label{eq:ising2-rzz}
\end{align}
The maximally mixed memory state is realized by equally averaging the four environmental preparations
$\ket{00},\ket{01},\ket{10},\ket{11}$, with the same preparation retained through both interaction intervals.

To access the process Choi operator, a maximally entangled pair is prepared for each temporal input. After the first interaction, the output system is stored by a SWAP operation and a fresh half of the second entangled pair is injected into the second interval. The four remaining Choi legs are then measured in the $X$ basis, directly yielding Eq.~\eqref{eq:ising2-witness}. The protocol is illustrated in Fig.~\ref{fig:ising2-experiment}.

\begin{figure}[t]
    \centering
    \includegraphics[width=\linewidth]{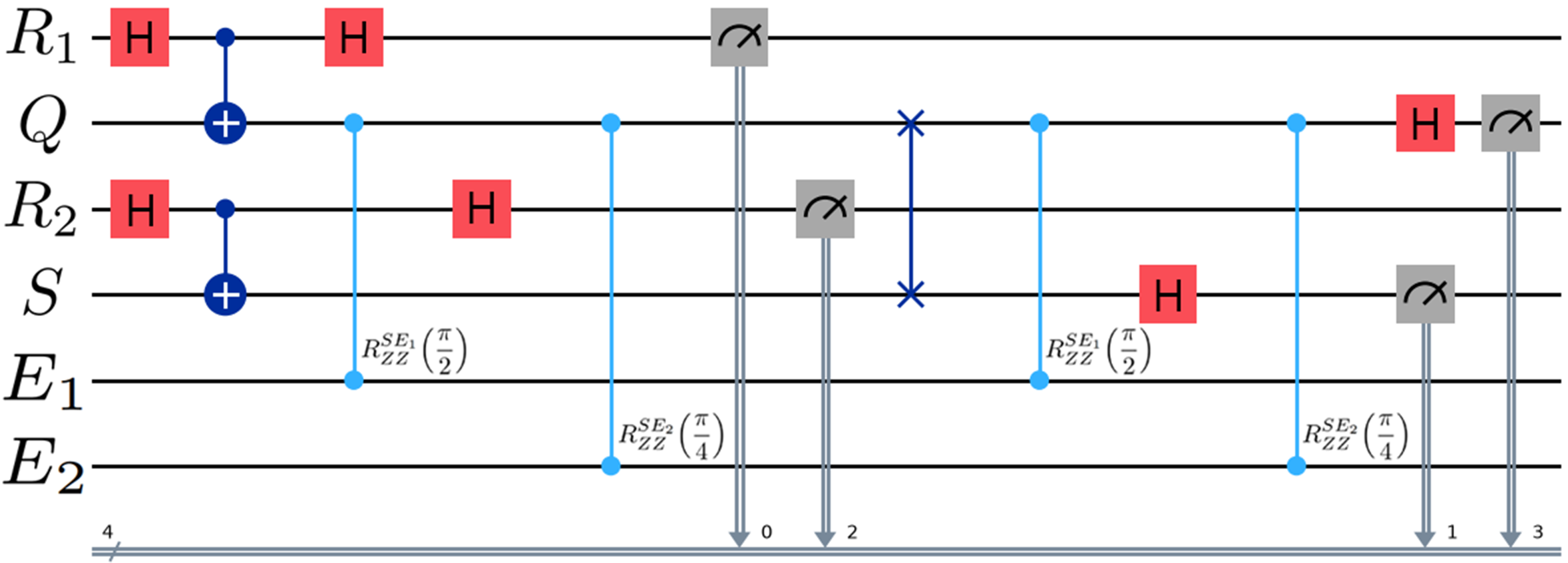}
    \caption{Ancilla-assisted measurement of the three-time process. Two Bell pairs retain the input and output Choi legs of the two consecutive intervals, while the same environmental qubits $E_1,E_2$ interact with the system in both intervals. A SWAP at the intermediate time stores the first output and injects the second Choi probe. The four visible process legs are measured in the $X$ basis by applying Hadamard gates followed by computational-basis measurements to evaluate the witness expectation value.}
    \label{fig:ising2-experiment}
\end{figure}

On the \texttt{ibm\_fez} quantum processor we obtain
\begin{align}
    |\gamma_{T_0\to T_1}|&=0.018\pm0.012,\nonumber\\
    |\gamma_{T_1\to T_2}|&=0.045\pm0.012,\nonumber\\
    |\gamma_{T_0\to T_2}|&=0.054\pm0.015,\nonumber\\
    |\sigma|&=0.3486\pm0.0146.
    \label{eq:ising2-qpu-results}
\end{align}
Despite increased hardware noise on the quantum processor, the multi-time witness remains $23.8$ shot-noise standard deviations above zero and retains a substantial fraction of its ideal value $0.5$, while the three channel diagnostics stay close to their ideal value of zero.

Full Choi tomography gives
\begin{align}
    \mathcal C_{\rel}(W_{\NM})=0.2689,
    \label{eq:ising2-experimental-process-coherence}
\end{align}
compared with the ideal value $\frac12\ln2\simeq0.3466$, while the reconstructed coherences of the three individual channels remain below $2.2\times10^{-2}$. The complete tomographic reconstruction and experimental details are given in Appendix~\ref{app:ising2-experiment}.

\section{Discussion and outlook}
\label{sec:discussion}

We introduced a resource theory of coherence for finite-dimensional multi-time processes, with incoherent processes defined by invariance under dephasing of all open legs. Relative entropy quantifies the coherence, while linear witnesses allow it to be detected without full process tomography.

Access to intermediate times can uncover coherence that is invisible in state- and channel-based descriptions. In the Ising model in which a system qubit interacts with two persistent environmental qubits, all three pairwise channels can be completely dephasing while the full three-time process remains coherent. The effect is carried by correlations between unitary histories generated by a persistent classical memory and does not require a quantum environment. On an IBM superconducting processor, both the four-leg witness and full process tomography recover this separation, despite residual hardware-induced channel coherence.

Possible extensions include continuously monitored and optomechanical open systems \cite{costa2026continuoustime}, where memory and multi-time access arise naturally, as well as extensions to continuous reference bases \cite{sajnok2026position}. A further direction is to explore coherence in the context of processes with indefinite causal order. Even though fully classical processes can have highly exotic causal structure, interpretable as a form of time travel \cite{Baumeler2017,Baumeler_2019, Tobar_2020}, coherence plays a prominent role in the physically relevant class of processes with quantum control of causal order \cite{Costa2020,Wechs2021,Grothus2025}, pointing to an intriguing connection between quantum resources and quantum causal structures.

\section*{Acknowledgments}
We acknowledge support from the International Network on Acausal Quantum Technologies (INAQT), supported by the EPSRC under grant EP/W026910/1, and from the Knut and Alice Wallenberg Foundation through the Wallenberg Initiative for Network and Quantum Information (WINQ). The authors have benefited from the activities of COST Action CA23115: Relativistic Quantum Information, funded by European Cooperation in Science and Technology. KS also acknowledges support from the 2026 Competition for Young Scientists of the Center for Theoretical Physics of the Polish Academy of Sciences and thanks Nordita for providing a summer internship opportunity. 

We acknowledge the use of IBM Quantum services for this work. The views expressed are those of the authors, and do not reflect the official policy or position of IBM or the IBM Quantum team.

The authors used ChatGPT (5.6 Sol, 6 Astra) and Claude (Opus 5.5) for literature synthesis, revision of the manuscript, and code development. The tools were used interactively on author-defined models, assumptions, equations, and tasks, with all outputs independently checked and corrected where necessary by the authors.

\bibliography{bibliography}

\appendix
\section{The Ising model with single environment qubit} \label{app:ising-model-1}

A single environmental qubit already gives a separation between fixed-origin channels and the full three-time process. In this model the channel between the two later times can still be coherent, so the separation is weaker than for the two-memory model of the main text. The same construction also extends directly to an arbitrary number of interaction intervals.

Let both the system $S$ and the environment $E$ be qubits, and consider
\begin{align}
    H^{ES}=Z^E\otimes Z^S, \label{eq:ising-hamiltonian}
\end{align}
where $Z=\ketbra{0}-\ketbra{1}$ is the third Pauli matrix. Writing the environmental $Z$ eigenvalue as $s=\pm1$, the unitary evolution operator is
\begin{align}
    U(t) = \sum_{s=\pm1}\ketbra{s}{s}^{E}\otimes U_s(t), \label{eq:ising-controlled-unitary}
\end{align}
where $U_s(t)=e^{-istZ}$. We take the balanced environmental state $\sigma^E=\frac{1}{2}\id_E$ unless stated otherwise. The label $s$ is conserved throughout the evolution and acts as a classical memory shared by all time intervals.

\paragraph{Reduced channel from a fixed initial time}

The reduced system channel after a time $t$ is
\begin{align}
    \cE_t(\rho) = \frac12\sum_{s=\pm1}U_s(t)\rho U_s^\dagger(t)=\cos^2t\,\rho+\sin^2t\,Z\rho Z. \label{eq:ising-reduced-channel}
\end{align}
Restricting to the nonzero support spanned by $\{\ket{00},\ket{11}\}$, normalized Choi state of Eq.~\eqref{eq:ising-reduced-channel} has the nonzero block
\begin{align}
    \widetilde{J}_{\cE_t}=\frac12 \begin{pmatrix} 1&\cos2t \\ \cos2t&1 \end{pmatrix}. \label{eq:ising-choi-block}
\end{align}
while the channel coherence in the $\ket{0}$, $\ket{1}$ basis (Fig.~\ref{fig:two-step-channel-coherence}) equals
\begin{align}
    C_{\rel}(\cE_t) = \ln2 - h_2\!\left(\frac{1+\cos2t}{2}\right), \label{eq:ising-channel-coherence}
\end{align}
where $h_2(q)\coloneqq-q\ln q-(1-q)\ln(1-q)$ is the binary entropy. Equation~\eqref{eq:ising-channel-coherence} vanishes exactly when $t=\frac{\pi}{4}+\frac{\pi n}{2}$, for $n\in\mathbb Z$, because then $\cE_t(\rho)=\frac12(\rho+Z\rho Z)=\D(\rho)$. Consequently, at these times every initial state is mapped to a state diagonal in the $Z$ basis, so both state coherence and channel coherence vanish.

\begin{figure}[t]
    \centering
    \includegraphics[width=\linewidth]{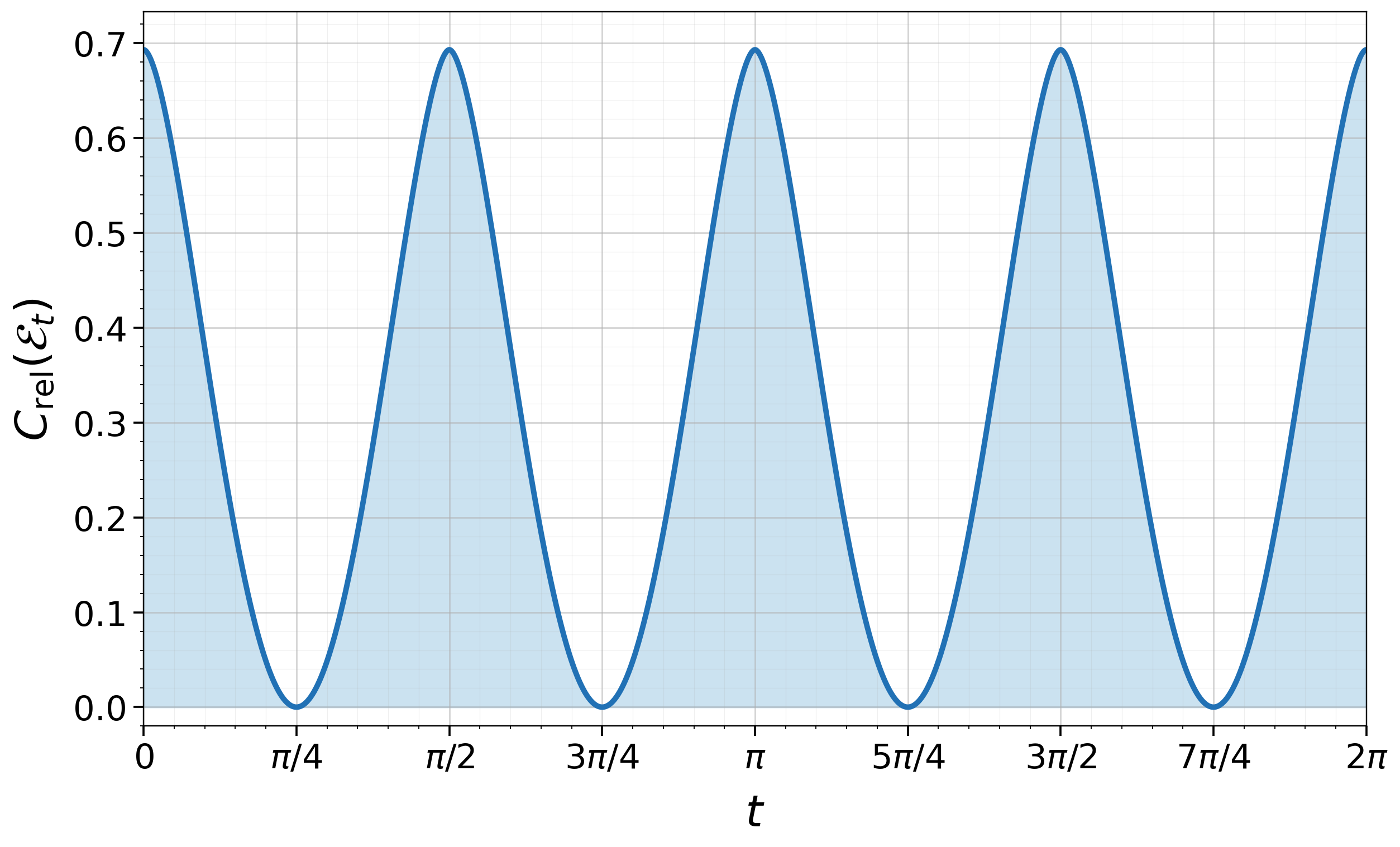}
    \caption{Relative entropy of coherence of the reduced fixed-origin channel $\cE_t$. The channel is completely dephasing at $t=\frac\pi4+n\frac\pi2$.}
    \label{fig:two-step-channel-coherence}
\end{figure}

\paragraph{The three-time process}

Let $\tau_1$ and $\tau_2$ denote the lengths of two consecutive interaction intervals, so that the absolute observation times are $T_1=\tau_1$, $T_2=\tau_1+\tau_2$. Keeping the same environmental label $s$ across both intervals and tracing out $E$ only after the second interaction gives the non-Markovian process
\begin{align}
    W_{\NM}^{A_OB_IB_OC_I} &= \frac12\sum_{s=\pm1} \PProj{U_s(\tau_1)}^{A_OB_I} \otimes \PProj{U_s(\tau_2)}^{B_OC_I}\!, \label{eq:nm-process-mixture}
\end{align}
where
\begin{align}
    \PProj{U_s(t)} = \sum_{\beta,\eta=0}^{1} \ee^{-\ii s t\left[(-1)^\beta-(-1)^\eta\right]} \ketbra{\beta}{\eta} \otimes \ketbra{\beta}{\eta},
\end{align}
that for $s=+1$ is the Choi operator of $U_{+}(t)=\ee^{-\ii tZ}$, while for $s=-1$ the Choi operator of $U_{-}(t)=\ee^{+\ii tZ}$.

The memory is classical, but it correlates the two conditional unitary histories. In the normalized nonzero support the process is
\begin{align}
    \widetilde{W}_{\NM} = \frac14 \begin{pmatrix}
        1&c_2&c_1&c_+\\
        c_2&1&c_-&c_1\\
        c_1&c_-&1&c_2\\
        c_+&c_1&c_2&1
    \end{pmatrix}, \label{eq:nm-block}
\end{align}
where $c_1=\cos2\tau_1$, $c_2=\cos2\tau_2$, and $c_\pm=\cos2(\tau_1\pm\tau_2)$. Its nonzero eigenvalues are  $\lambda_\pm = \frac12(1\pm c_1c_2)$, while $\cD\!\left(\widetilde{W}_{\NM}\right)=\frac14\id_4$. We get
\begin{align}
    \mathcal{C}_{\rel}(W_{\NM}) = \ln4- h_2\!\left(\frac{1+c_1c_2}{2}\right), \label{eq:nm-coherence}
\end{align}
that satisfies $\mathcal{C}_{\rel}\ge\ln2$ for all $\tau_1,\tau_2$, implying a finite lower bound on the coherence of the multi-time process in the chosen basis; see Fig.~\ref{fig:ising_multitime_coherence}.

The channel representation of Eq.~\eqref{eq:nm-process-mixture} makes the role of the shared memory explicit. For independent inputs $\rho_1$ on $A_O$ and $\rho_2$ on $B_O$,
\begin{align}
    \cW_{\tau_1,\tau_2}(\rho_1\otimes\rho_2) =& \cE_{\tau_1}(\rho_1)\otimes\cE_{\tau_2}(\rho_2) \nonumber\\
    &-\frac14\sin(2\tau_1)\sin(2\tau_2) [Z,\rho_1]\otimes[Z,\rho_2]. \label{eq:ising-process-channel}
\end{align}
The second term is absent from the product of the two reduced channels. It arises because the same sign $s$ multiplies the coherent part of both conditional unitary rotations.

\paragraph{Advantage over fixed-origin channels} \label{sec:fixed-origin-advantage}

In the standard open-system description the initial time is fixed at $T_0=0$, and the dynamics is characterized through states or channels from $T_0$ to later observation times. For the present model the channel from $T_0$ to $T_1$ is $\cE_{\tau_1}$. If no intervention is performed at $T_1$, the same environmental label persists and the channel from $T_0$ to $T_2$ is
\begin{align}
    \cE_{0\to T_2}(\rho) &= \frac12\sum_{s=\pm1} U_s(\tau_2)U_s(\tau_1) \rho U_s^\dagger(\tau_1)U_s^\dagger(\tau_2). \label{eq:fixed-origin-second-channel}
\end{align}
Choose $T_1=\frac{\pi}{4}$, $T_2=\frac{3\pi}{4}$. Both fixed-origin channels are then completely dephasing:
\begin{align}
    \cE_{0\to T_1}&=\D= \cE_{0\to T_2}\\
    C_{\rel}(\cE_{0\to T_1}) &= 0 = C_{\rel}(\cE_{0\to T_2}). \label{eq:state-channel-zero}
\end{align}

The full three-time process remains coherent. Substituting $\tau_1=T_1=\frac\pi4$, $\tau_2=T_2-T_1=\frac\pi2$ into Eq.~\eqref{eq:nm-block} gives
\begin{align}
    \widetilde{W}_{\NM} = \frac14 \begin{pmatrix}
        1&0&-1&0\\
        0&1&0&-1\\
        -1&0&1&0\\
        0&-1&0&1
    \end{pmatrix}\; \Rightarrow \; \mathcal{C}_{\rel}(W_{\NM})=\ln2.
\label{eq:main-separation-process}
\end{align}

\paragraph{Comparison with the Markovian product process} \label{app:markovian-comparison}

Resetting the environment after each interval gives the Markovian process
\begin{align}
    W_{\M} = J_{\cE_{\tau_1}}^{A_OB_I} \otimes J_{\cE_{\tau_2}}^{B_OC_I}. \label{eq:markovian-product}
\end{align}
For normalized Choi operators, relative-entropy coherence is additive on tensor products, so
\begin{align}
    \mathcal{C}_{\rel}(W_{\M}) = C_{\rel}(\cE_{\tau_1}) + C_{\rel}(\cE_{\tau_2}), \label{eq:markovian-coherence}
\end{align}
which is plotted in Fig.~\ref{fig:ising_multitime_coherence}. At $\tau_1=\frac\pi4=\tau_2$ this gives $\mathcal{C}_{\rel}(W_{\M})=0$, whereas Eq.~\eqref{eq:main-separation-process} gives $\mathcal{C}_{\rel}(W_{\NM})=\ln2$. The difference is generated by the persistent classical memory label shared by the two conditional unitary histories.

\begin{figure}[t]
    \centering
    \begin{overpic}[width=\linewidth]{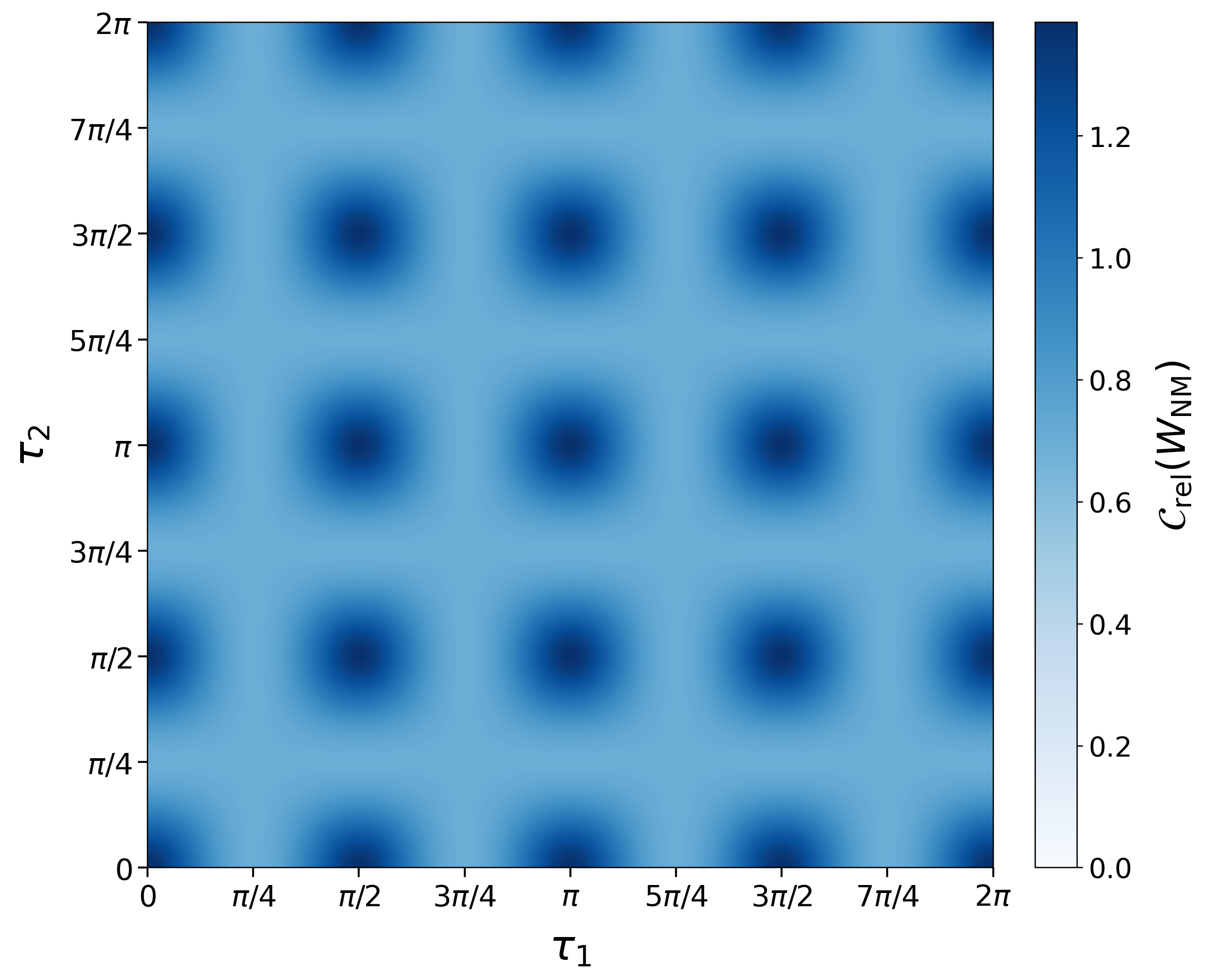}
        \put(1,78){(a)}
    \end{overpic}
    \vspace{0.5em}
    \begin{overpic}[width=\linewidth]{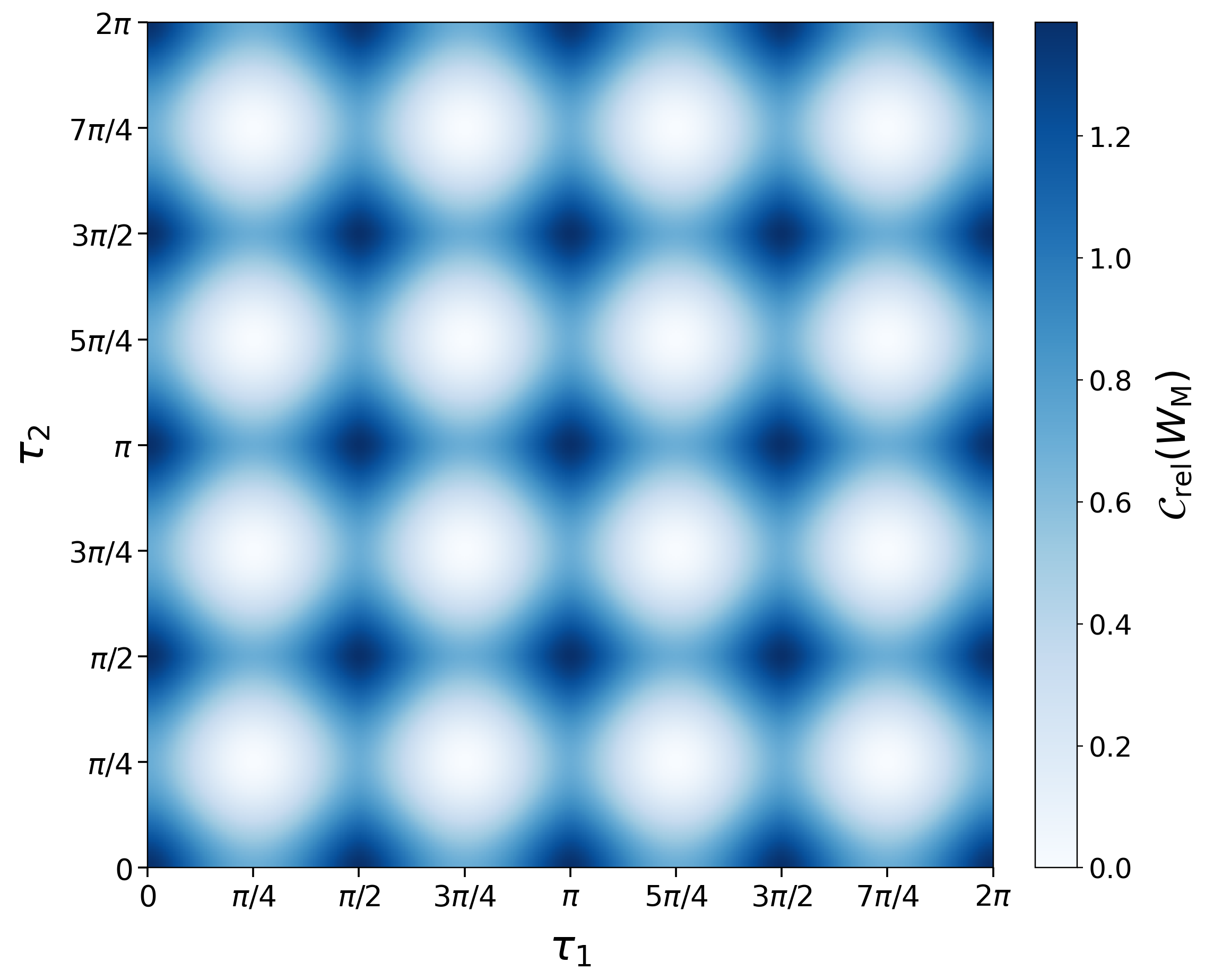}
        \put(1,78){(b)}
    \end{overpic}
    \caption{Relative entropy of coherence of the three-time Ising process with a single environmental qubit as a function of the two interval lengths: (a) the non-Markovian process, $\mathcal{C}_{\rel}(W_{\NM})$, retains at least $\ln 2$ of coherence; (b) the Markovian coherence, $\mathcal{C}_{\rel}(W_{\M})$, is additive over the two intervals and vanishes at $\tau_1=\frac\pi4=\tau_2$.}
    \label{fig:ising_multitime_coherence}
\end{figure}

\paragraph{Scaling to an arbitrary number of time intervals}
\label{sec:n-time-scaling}

Consider $N$ consecutive interaction intervals of lengths $\tau_1,\ldots,\tau_N$, while keeping the same environmental label $s=\pm1$ throughout the entire evolution. The corresponding multi-time process is
\begin{align}
    W_{\NM}^{(N)} = \frac12\sum_{s=\pm1} \bigotimes_{k=1}^{N} \PProj{U_s(\tau_k)}. \label{eq:N-time-process}
\end{align}
After normalization, this is an equal mixture of two pure process histories,
\begin{align}
    \widetilde W_{\NM}^{(N)} = \frac12 \left( \ketbra{\Psi_+}{\Psi_+} + \ketbra{\Psi_-}{\Psi_-} \right),
\end{align}
where $\ket{\Psi_s} = \bigotimes_{k=1}^{N} \frac{1}{\sqrt2} \left( \ee^{-\ii s\tau_k}\ket{00} + \ee^{+\ii s\tau_k}\ket{11} \right)$. The overlap between the two histories factorizes as
\begin{align}
    \braket{\Psi_+}{\Psi_-} = \prod_{k=1}^{N}\cos 2\tau_k,
\end{align}
and the only nonzero eigenvalues of $\widetilde W_{\NM}^{(N)}$ are
\begin{align}
    \lambda_\pm^{(N)} = \frac12 \left( 1\pm\prod_{k=1}^{N}\cos 2\tau_k \right). \label{eq:N-time-eigenvalues}
\end{align}

The process coherence takes the closed form
\begin{align}
    \mathcal C_{\rel}\!\left(W_{\NM}^{(N)}\right) = N\ln2 - h_2\!\left[ \frac12 \left( 1+\prod_{k=1}^{N}\cos 2\tau_k \right) \right], \label{eq:N-time-coherence}
\end{align}
that with $h_2(q)\leq\ln2$ yields 
\begin{align}
    \mathcal C_{\rel}\!\left(W_{\NM}^{(N)}\right) \geq (N-1)\ln2. \label{eq:N-time-lower-bound}
\end{align}
The minimum coherence of this family of multi-time processes grows linearly with the number of interaction intervals.

The separation from the fixed-origin channels is also linear in $N$. Let $\tau_1=\frac{\pi}{4}$, $\tau_k=\frac{\pi}{2}$ for $ k=2,\ldots,N$. The absolute observation times are then $T_j = \frac{\pi}{4} + (j-1)\frac{\pi}{2}$ for $j=1,\ldots,N$. For every such time, $\cos 2T_j=0$, so every fixed-origin reduced channel is completely dephasing,
\begin{align}
    \cE_{0\rightarrow T_j} = \Delta \; \Rightarrow \; C_{\rel}\!\left(\cE_{0\rightarrow T_j}\right)=0 \label{eq:N-time-fixed-origin}
\end{align}
However, because $\cos 2\tau_1=0$, Eq.~\eqref{eq:N-time-coherence} gives
\begin{align}
    \mathcal C_{\rel}\!\left(W_{\NM}^{(N)}\right) = (N-1)\ln2. \label{eq:N-time-separation}
\end{align}
An arbitrarily long sequence of fixed-origin observations can appear completely incoherent, while the underlying multi-time process contains an amount of coherence that grows linearly with the number of temporal intervals.

\section{The dipole-dipole model} \label{app:dipole-model}

As a physically-motivated model, consider the magnetic dipole-dipole interaction for an interspin axis chosen along $z$. Let both the system $S$ and the environment $E$ be qubits, initially uncorrelated, with the environment in the maximally mixed state $\sigma^E=\frac12\id_E$. We consider the interaction
\begin{align}
    H^{ES} = J\left(2Z^E Z^S-X^E X^S-Y^E Y^S \right), \label{eq:dipole-hamiltonian}
\end{align}
where $X,Y,Z$ are the Pauli matrices. In contrast to the Ising interaction, the environmental $Z$ eigenvalue is not conserved: the $XX+YY$ part exchanges an excitation between $S$ and $E$, so the environment serves as a quantum memory rather than a fixed classical label.

The propagator of Eq.~\eqref{eq:dipole-hamiltonian} equals
\begin{align}
    U(t) =& \ee^{-2\ii Jt} \left(\ketbra{00}{00}+\ketbra{11}{11}\right) \nonumber\\
    &+ \ee^{+2\ii Jt}\cos{2Jt} \left(\ketbra{01}{01}+\ketbra{10}{10} \right) \nonumber\\
    &+ \ii\ee^{+2\ii Jt}\sin{2Jt} \left( \ketbra{01}{10} + \ketbra{10}{01} \right).
    \label{eq:dipole-propagator}
\end{align}

\paragraph{Channels}
Tracing out the maximally mixed environment gives a unital Pauli channel. The reduced evolution is
\begin{align}
    \cE_t(\rho) = \frac12\Big[ &\id  + \lambda_\perp(t)  \left( \tr(X\rho) X+\tr(Y\rho)Y \right) \nonumber \\
    &+ \lambda_z(t)\tr(Z\rho)zZ \Big], \label{eq:dipole-reduced-channel}
\end{align}
where
\begin{align}
    \lambda_\perp(t) &= \cos(2Jt)\cos(4Jt), \label{eq:dipole-lambda-perp}\\
    \lambda_z(t) &= \cos^2(2Jt). \label{eq:dipole-lambda-z}
\end{align}
The normalized Choi state reads, in the computational basis
$\{\ket{00},\ket{01},\ket{10},\ket{11}\}$,
\begin{align}
    \widetilde J_{\cE_t} = \frac14
    \begin{pmatrix}
        1+\lambda_z & 0 & 0 & 2\lambda_\perp\\
        0 & 1-\lambda_z & 0 & 0\\
        0 & 0 & 1-\lambda_z & 0\\
        2\lambda_\perp & 0 & 0 & 1+\lambda_z
    \end{pmatrix},
    \label{eq:dipole-choi}
\end{align}
and the channel is incoherent in the $Z$ reference basis whenever
$\lambda_\perp(t)=0$. Its relative entropy of coherence is
\begin{align}
    C_{\rel}(\cE_t)  = \frac{1+\lambda_z(t)}{2} \left(  \ln2  - h_2\!\left[ \frac12+ \frac{|\lambda_\perp(t)|}{1+\lambda_z(t)} \right] \right).\label{eq:dipole-channel-coherence}
\end{align}
At $\tau=\frac{\pi}{8J}$,
\begin{align}
    C_{\rel}(\cE_\tau)=0, \label{eq:dipole-first-channel}
\end{align}
and $\cE_\tau$ maps every input state to a state diagonal in the $Z$ basis,
although it is not the pure dephasing map,
\begin{align}
    \cE_\tau(\rho) = \frac12\id + \frac14\Tr(Z\rho)Z. \label{eq:dipole-tau-channel}
\end{align}
After twice this interaction time, $\lambda_\perp(2\tau) = \lambda_z(2\tau) = 0$, and the fixed-origin channel becomes completely depolarizing,
\begin{align}
    \cE_{2\tau}(\rho) = \frac12\id, \qquad C_{\rel}(\cE_{2\tau})=0. \label{eq:dipole-double-channel}
\end{align}

Consider the three observation times $ T_0=0$,  $T_1=\tau$, $T_2=2\tau$. The channel from $T_0$ to $T_1$ is $\cE_\tau$, while in the absence of an intervention at $T_1$ the channel from $T_0$ to $T_2$ is $\cE_{2\tau}$. Both are incoherent.

For a non-Markovian process, a channel associated with the second interval is not defined independently of how the earlier part of the process is closed. We use the incoherent closure: prepare the maximally mixed state $\frac12\id$ at $A_O$ and discard the system at $B_I$. Since
\begin{align}
    U(\tau) \left( \frac12\id_S\otimes\frac12\id_E \right) U^\dagger(\tau) = \frac12\id_S\otimes\frac12\id_E,
\end{align}
the environment entering the second interval is again maximally mixed and the induced channel from $T_1$ to $T_2$ is also the incoherent $\cE_\tau$. 

The relative entropy of coherence vanishes for all three channels $\cE_{0\to T_1}$, $\cE_{0\to T_2}$ and $\cE_{T_1\to T_2}$.

\paragraph{The three-time process}

Let $\tau_1$ and $\tau_2$ denote two consecutive interaction intervals. Define
operators acting on $S$ by
\begin{align}
    M_{mn}(t) \coloneqq {}_E\!\bra{m}U(t)\ket{n}_E.
    \label{eq:dipole-M-definition}
\end{align}
Equation~\eqref{eq:dipole-propagator} gives
\begin{align}
    M_{00}(t) &= \ee^{-2\ii Jt}\ketbra{0}{0} + \ee^{+2\ii Jt}\cos2Jt\,\ketbra{1}{1}, \nonumber\\
    M_{11}(t) &= \ee^{+2\ii Jt}\cos 2Jt\,\ketbra{0}{0} + \ee^{-2\ii Jt}\ketbra{1}{1}, \nonumber\\
    M_{01}(t) &= \ii\ee^{+2\ii Jt}\sin 2Jt\,\ketbra{1}{0}, \nonumber\\
    M_{10}(t) &= \ii\ee^{+2\ii Jt}\sin 2Jt\,\ketbra{0}{1}.
    \label{eq:dipole-M-operators}
\end{align}
Define the two-interval history vectors
\begin{align}
    \kket{\cM_{mn}} = \sum_{k=0}^{1} \kket{M_{kn}(\tau_1)}^{A_OB_I} \otimes \kket{M_{mk}(\tau_2)}^{B_OC_I}, \label{eq:dipole-history-vectors}
\end{align}
where $\kket{M} = (\id\otimes M) \left(\ket{00}+\ket{11}\right)$. The process matrix is then
\begin{align}
    W_{\NM}^{A_OB_IB_OC_I} = \frac12 \sum_{m,n=0}^{1} \kket{\cM_{mn}} \bbra{\cM_{mn}}.
    \label{eq:dipole-process-general}
\end{align}
The sum over the intermediate environmental state $k$ in
Eq.~\eqref{eq:dipole-history-vectors} is coherent. This is the structural
difference from a process generated by a conserved classical environmental
label, for which the different memory histories appear only as a classical
mixture.

For equal intervals, $\tau_1=\tau_2=\tau=\frac{\pi}{8J}$. The normalized process state is $\omega_{\NM}=\frac{W_{\NM}}4$ and the relative entropy of process coherence equals
\begin{align}
    \mathcal C_{\rel}(W_{\NM}) = \frac14\ln2 + \frac38\ln3 + \frac5{16}\ln5.
    \label{eq:dipole-process-coherence}
\end{align}

For the three pairwise characterizations specified above, the channel coherence vanishes, while the complete three-time process remains coherent.

\section{Experimental implementation and tomography of the two-memory Ising model}
\label{app:ising2-experiment}

We implemented the two-memory Ising model of Sec.~\ref{sec:ising-model-2} numerically and on IBM hardware. The Hamiltonian, Eq.~\eqref{eq:ising2-hamiltonian} is simulated for $\tau=\frac{\pi}{8g}$ using
\begin{align}
    U_\tau = R_{ZZ}^{SE_1}\!\left(\frac{\pi}{2}\right) R_{ZZ}^{SE_2}\!\left(\frac{\pi}{4}\right).
\end{align}
The maximally mixed state of the two memory qubits is implemented as the equal classical mixture $\frac14\id_{E_1E_2}$. Each observable is evaluated for the four computational-basis preparations of $E_1E_2$ and then averaged. Within a single circuit the memory qubits are not reset between the two interaction intervals.

\paragraph{Ancilla-assisted channel measurements}

For each reduced channel a Bell pair $\ket{\Phi^+} = \frac{1}{\sqrt2} \left(\ket{00}+\ket{11}\right)$ is prepared between the system and a reference qubit. Acting with the channel on the system half produces its normalized Choi state. For the pure-dephasing family considered here,
\begin{align}
    \widetilde J_\cE = \frac12 \begin{pmatrix}
        1&\gamma\\
        \gamma^*&1
    \end{pmatrix}
\end{align}
on its nonzero support, so that measurements in Eq.~\eqref{eq:ising2-channel-diagnostic} are performed for the channel-coherence diagnostic.

For the second interval, the first system output is discarded and a fresh half of a Bell pair is injected through a SWAP operation. Since the memory populations are unchanged by the first interaction, this realizes the same channel $\cE_\tau$ independently of the initial system preparation.

\paragraph{Ancilla-assisted process measurement}

The three-time Choi state is generated coherently using two Bell pairs. Let $R_1$ be the reference associated with the first temporal input and let $R_2$ be the reference associated with the second. A further qubit $Q$ stores the output of the first interval. The circuit implements
\begin{align}
    \ket{\Phi^+}_{R_1S}
    \xrightarrow{\,U_\tau\,}
    \operatorname{SWAP}_{SQ}
    \xrightarrow{\,U_\tau\,}||,
\end{align}
where before the SWAP the pair $\ket{\Phi^+}_{R_2Q}$ is prepared. After the SWAP, the four qubits $(R_1,Q,R_2,S)$ represent the process legs $(A_O,B_I,B_O,C_I)$, respectively. Measuring all four in the $X$ basis yields the witness in Eqs.~\eqref{eq:ising2-witness}-\eqref{eq:ising2-witness-value}.

The targeted measurements used $1024$ shots for each environmental branch and measurement setting. The Qiskit implementation was developed with assistance from ChatGPT (5.6 Sol) and verified by the authors. We first tested the circuits on the noisy \texttt{FakeFez} backend and then ran them on \texttt{ibm\_fez} quantum processor. The results are summarized in Table~\ref{tab:ising2-results}.

The witness reaches $94\%$ of its ideal value in the noisy simulation and $70\%$ on hardware; for the reconstructed process coherence the corresponding values are $95\%$ and $78\%$. In both cases, the process coherence exceeds the residual coherence of any individual channel by more than an order of magnitude.

\begin{table}[t]
\centering
\caption{Comparison of the ideal predictions, noisy IBM simulation, and results obtained on the \texttt{ibm\_fez} quantum processor.}
\label{tab:ising2-results}
\begin{tabular}{l|c|c|c}
\hline\hline
Quantity & Ideal & FakeFez & \texttt{ibm\_fez} \\
\hline
$|\gamma_{01}|$ & $0$ & $0.013 \pm 0.011$ & $0.018 \pm 0.012$ \\
$|\gamma_{12}|$ & $0$ & $0.015 \pm 0.012$ & $0.045 \pm 0.012$ \\
$|\gamma_{02}|$ & $0$ & $0.014 \pm 0.013$ & $0.054 \pm 0.015$ \\
$|\sigma|$ & $0.50$ & $0.47 \pm 0.014$ & $0.35 \pm 0.015$ \\
\hline
$C_{\rel}(\cE_{01})$ & $0$ & $0.010$ & $0.018$ \\
$C_{\rel}(\cE_{12})$ & $0$ & $0.012$ & $0.013$ \\
$C_{\rel}(\cE_{02})$ & $0$ & $0.0015$ & $0.022$ \\
$\mathcal C_{\rel}(W)$ & $0.35$ & $0.33$ & $0.27$ \\
\hline\hline
\end{tabular}
\end{table}

\begin{figure}[t]
    \centering
    \includegraphics[width=\linewidth]{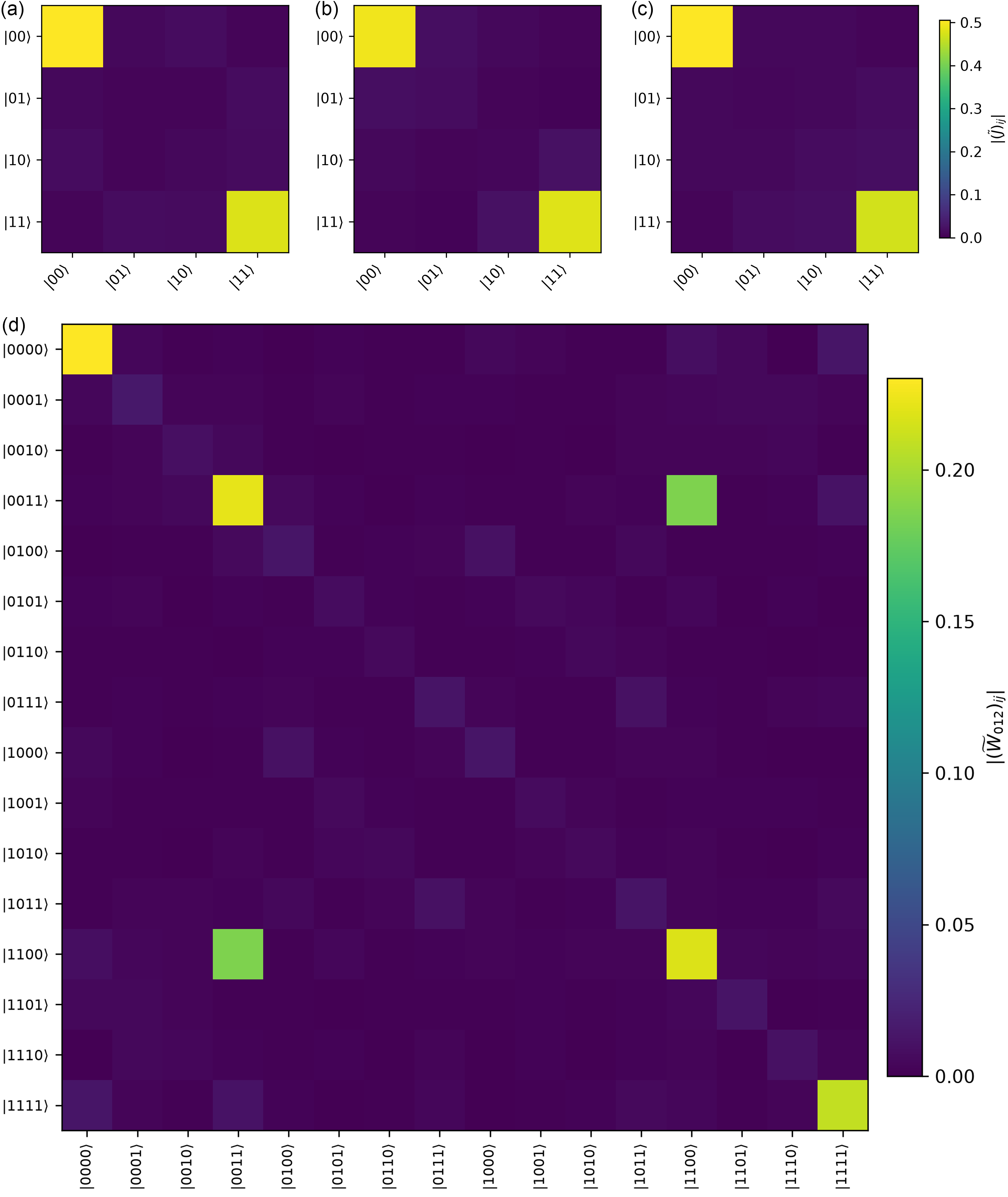}
    \caption{Tomographic reconstruction of the reduced channels and the full three-time process on \texttt{ibm\_fez}. Panels (a)--(c) show the absolute values of the normalized Choi matrices of the three pairwise channels, $\cE_{T_0\rightarrow T_1}$, $\cE_{T_1\rightarrow T_2}$, and $\cE_{T_0\rightarrow T_2}$, respectively, which remain close to diagonal in the reference basis. Panel (d) shows the absolute values of the normalized process Choi matrix, exhibiting sizeable off-diagonal elements associated with multi-time coherence.}
    \label{fig:tomo}
\end{figure}

\paragraph{Full tomography}

For the channel Choi states we perform local Pauli tomography using all $3^2=9$ settings in the $X$, $Y$, and $Z$ bases. For the four-qubit process Choi state we analogously use $3^4=81$ local Pauli settings. Each setting is evaluated for all four environmental branches.

The density operators are reconstructed by linear inversion. Finite-shot statistics and noise can produce small negative eigenvalues; before evaluating the von Neumann entropy we set negative eigenvalues to zero and renormalize. The reported tomographic values are visualized in Fig~\ref{fig:tomo}. We do not report uncertainties for the relative entropies because each tomographic reconstruction was performed only once. A practical advantage of the framework is that process coherence can be certified with a witness, without full process tomography.

\end{document}